\documentclass[sn-mathphys-num]{sn-jnl}

\usepackage[utf8]{inputenc}
\usepackage[ruled,linesnumbered,noend]{algorithm2e}
\usepackage{subcaption}

\usepackage{adjustbox}
\usepackage{booktabs}
\usepackage{pgfplotstable}
\pgfplotsset{compat=1.18}

\usepackage{graphicx}%
\usepackage{multirow}%
\usepackage{amsmath,amssymb,amsfonts}%
\usepackage{amsthm}%
\usepackage{mathrsfs}%
\usepackage[title]{appendix}%
\usepackage{xcolor}%
\usepackage{textcomp}%
\usepackage{manyfoot}%
\usepackage{booktabs}%

\newcommand{\set}[1]{\left\{#1\right\}}

\newcommand{\fpr}[1]{\mathopen{}\left(#1\right)}

\newcommand{\fspr}[1]{\mathopen{}\left[#1\right]}

\newcommand{\abs}[1]{{\left|#1\right|}}
\newcommand{\norm}[1]{\left\|#1\right\|}
\newcommand{\ceil}[1]{{\left\lceil#1\right\rceil}}

\newcommand{\real}{\mathbb{R}}
\newcommand{\np}{\textbf{NP}}
\newcommand{\apx}{\textbf{APX}}

\newcommand{\funcdef}[3]{{#1}:{#2} \to {#3}}

\newcommand{\opt}{\mathit{OPT}}

\newcommand{\fm}[1]{{\mathcal{#1}}}

\DeclareRobustCommand{\dispfunc}[2]{%
    \ensuremath{%
        \ifthenelse{\equal{#2}{}}%
            {\mathit{#1}}%
            {\mathit{#1}\fpr{#2}}}}

\newcommand{\bigO}[1]{\dispfunc{\mathcal{O}}{#1}}
\newcommand{\bigOstar}[1]{\dispfunc{\mathcal{O}^*}{#1}}

\newcommand{\prbgts}{\textsc{dgs}\xspace}
\newcommand{\prbgtsc}{\textsc{dgs-centroid}\xspace}
\newcommand{\prbgtss}{\textsc{dgs-sort}\xspace}
\newcommand{\prbgtsp}{\textsc{dgs-partition-ne}\xspace}
\newcommand{\prbgtspalt}{\textsc{dgs-partition}\xspace}
\newcommand{\prbmtc}{\textsc{mtc}\xspace}
\newcommand{\prbfas}{\textsc{fas}\xspace}

\newcommand{\alggreedy}{\textsc{Greedy}\xspace}
\newcommand{\algcut}{\textsc{Mcut}\xspace}
\newcommand{\algdp}{\textsc{TreeDP}\xspace}
\newcommand{\alglpiter}{\textsc{LPiter}\xspace}
\newcommand{\alglpsym}{\textsc{LPsym}\xspace}
\newcommand{\algexact}{\textsc{Exact}\xspace}
\newcommand{\algilpiter}{\textsc{ILPiter}\xspace}
\newcommand{\algkmeans}{\textsc{KMeans}\xspace}
\newcommand{\algkcut}{\textsc{KCut}\xspace}
\newcommand{\algrandom}{\textsc{Random}\xspace}

\newcommand{\dtname}[1]{\textsl{#1}}
\newcommand{\dtenron}{\dtname{Enron}\xspace}
\newcommand{\dtdblp}{\dtname{DBLP}\xspace}
\newcommand{\dtreddit}{\dtname{Reddit\_thread}\xspace}
\newcommand{\dtreddithyperlink}{\dtname{Reddit\_hyperlink}\xspace}
\newcommand{\dttwitter}{\dtname{Twitter}\xspace}
\newcommand{\dtwiki}{\dtname{Wikipedia}\xspace}

\renewcommand{\refname}{References}
\makeatletter
\renewcommand\bibsection{%
  \section*{{\refname}\@mkboth{\refname}{\refname}}%
}%
\makeatother

\theoremstyle{thmstyleone}%
\newtheorem{theorem}{Theorem}
\newtheorem{problem}{Problem}
\newtheorem{lemma}{Lemma}
\newtheorem{corollary}{Corollary}

\theoremstyle{thmstyletwo}%

\theoremstyle{thmstylethree}%

\begin{document}

\title{Finding coherent node groups in directed graphs}

\author*[1]{\fnm{Iiro} \sur{Kumpulainen}}\email{iiro.kumpulainen@helsinki.fi}

\author*[1]{\fnm{Nikolaj} \sur{Tatti}}\email{nikolaj.tatti@helsinki.fi}

\affil[1]{\orgname{University of Helsinki, HIIT}, \city{Helsinki}, \country{Finland}}

\abstract{Grouping the nodes of a graph into clusters is a standard technique for studying networks. We study a problem where we are given a directed network and are asked to partition the graph into a sequence of coherent groups. We assume that nodes in the network have features, and we measure the group coherence by comparing these features. Furthermore, we incorporate the cross edges by penalizing the forward cross edges and backward cross edges with different weights. If the weights are set to 0, then the problem is equivalent to clustering. However, if we penalize the backward edges, the order of discovered groups matters, and we can view our problem as a generalization of a classic segmentation problem.

We consider a common iterative approach where we solve the groups given the centroids, and then find the centroids given the groups. We show that---unlike in clustering---the first subproblem is \textbf{NP}-hard. However, we show that we can solve the subproblem exactly if the underlying graph is a tree or if the number of groups is 2. For a general case, we propose an approximation algorithm based on linear programming.

We propose 3 additional heuristics:
(1) optimizing each pair of groups separately while keeping the remaining groups intact, (2) computing a spanning tree and then optimizing using only the edges in that, and (3) a greedy search moving nodes between the groups while optimizing the overall loss. We demonstrate with our experiments that the algorithms are practical and yield interpretable results.}


\maketitle

\section{Introduction}

Summarizing a large graph by grouping the nodes into clusters is a standard technique for studying networks. While many techniques have been proposed for clustering undirected graphs, directed graphs pose additional challenges.

On the other hand, much data can be naturally represented using directed networks, such as discussion threads in social media platforms or a citation graph. In addition to edges we also typically have additional information attached to the nodes, typically expressed as categorical labels or real-valued features. These features allow us to measure the similarity of the nodes, which in turn allows us to cluster similar nodes together.
When clustering nodes, we would like to take edges into account. For example, given a citation graph, our goal is to partition it into groups of similar nodes such that one group cites the other. Another example is a discussion thread where our goal is to group early messages in one group and replies (and the following replies) in the other group.

We consider discovering ordered partitions in a directed graph.
That is, given a graph, our goal is to divide the vertices in a \emph{sequence} of $k$ groups such that each group is as coherent as possible while (backward) cross edges are minimized.

We focus on using $L_2$ loss, though our approach will work on any centroid-based objective. We will refer to this problem as \prbgts. The \np-hardness of clustering immediately implies the hardness of \prbgts. In addition, we prove that \prbgts is \apx-hard by reducing the problem from the feedback arc set problem. 

To solve \prbgts, we first consider an exact solver based on mixed integer linear programming. Unfortunately, this is usable only for small cases, and therefore, we consider several heuristics.

First, we consider a greedy search where we decrease the cost by moving the vertices from one cluster to another. We show that by using a common $L_2$ decomposition we can run a single iteration in $\bigO{k(nd + m)}$ time, where $n$ and $m$ are the numbers of nodes and edges, and $d$ is the number of features.

We also propose an iterative approach where we fix centroids and optimize the partition, and then fix the partition and optimize centroids. We refer to the first problem as \prbgtspalt.
Unlike with the $k$-means algorithm, solving \prbgtspalt is \np-hard.

We then consider two common special cases. We show that if the input graph is a tree, we can solve \prbgtspalt with dynamic programming in $\bigO{dn}$ time. We also show that if $k = 2$, we can find the partition with a minimum cut in $\bigO{n(d + m)}$ time. For a general case, we propose an algorithm that enumerates all pairs of groups and optimizes them using a minimum cut while keeping the remaining groups fixed.

In addition, we propose two linear program approaches for \prbgtspalt: one for the general case, and one for a special case where we penalize backward cross edges and forward cross edges equally. We show that the former yields a $k - 1$ approximation guarantee while the latter yields a $\frac{k + 1}{3}$ approximation guarantee.

A summary of the algorithms and their running times are given in Table~\ref{tab:algo}.

This manuscript is an extension of the previously published work~\citep{kumpulainen2023finding}. Our main extension of the previous work is the introduction of the linear program approaches in Section~\ref{sec:milp} and Section~\ref{sec:lp}, as well as more comprehensive experiments in Section~\ref{sec:exp}. 

The remainder of the paper is organized as follows. We present preliminary notation and define the problem formally in Section~\ref{sec:prel}. We describe the related work in Section~\ref{sec:related}.
We continue by describing an exact MILP solver for \prbgts in Section~\ref{sec:milp}.
Next, we consider an iterative algorithm in Section~\ref{sec:iterative}, show the computational complexity of the main problem and the related sub-problems in Section~\ref{sec:complexity}, discuss the special cases of \prbgtspalt in Section~\ref{sec:sub}, and describe the linear program approaches in Section~\ref{sec:lp}.
Additional algorithms for solving \prbgts are given in Section~\ref{sec:algo}.
We present our experiments in Section~\ref{sec:exp}, and conclude the paper with a discussion in Section~\ref{sec:conclusions}.

\begin{table}

\caption{Summary of the algorithms. Here, $n$ and $m$ are the number of nodes and edges in the input graph, $d$ is the number of features, $k$ is the number of clusters. The functions $\mathit{LP}(X)$, $\mathit{ILP}(X)$, and $\mathit{MILP}(X)$ are the running times of the respective programs with at most $X$ variables and constraints.
The function $\mathit{FAS(k)}$ is the running time needed to solve the feedback arc set problem for a graph with $k$ nodes. This is solved either exactly with ILP in exponential time, or with a heuristic in $\bigO{k^2}$ time. All times are per iteration, except for \algexact.
}
\label{tab:algo}

\begin{tabular}{lrp{4.4cm}}
\toprule
Algorithm & running time & notes \\
\midrule
\algexact & $\mathit{MILP}(\bigO{n^2 + k(m + n)}) + \bigO{n^2d}$ & exact  \\
\alggreedy & $\bigO{k(nd + m)} + \mathit{FAS}(k)$ & iterative \\
\algdp & $\bigO{knd + m} + \mathit{FAS}(k)$ & iterative, solves \prbgtspalt exactly when graph is a tree \\
\algcut & $\bigO{k^2n(d + m)} + \mathit{FAS}(k)$ & iterative, solves \prbgtspalt exactly when $k = 2$ \\
\alglpiter & $\mathit{LP}(\bigO{k(n + m)}) + \bigO{k(nd + m)} + \mathit{FAS}(k)$ & iterative, yields $k - 1$ approximation for \prbgtspalt \\
\alglpsym & $\mathit{LP}(\bigO{k(n + m)})  + \bigO{k(nd + m)}$ & iterative, assumes equal penalties for cross edges, yields $\frac{k + 1}{3}$ approximation for \prbgtspalt \\
\algilpiter & $\mathit{ILP}(\bigO{k(n + m)}) + \bigO{k(nd + m)} + \mathit{FAS}(k)$ & iterative, solves \prbgtspalt with non-empty clusters exactly \\
\bottomrule
\end{tabular}
\end{table}
\section{Preliminary notation and problem definition}\label{sec:prel}
We begin by establishing the notation that we will use throughout the paper and by defining our main optimization problem.

We assume that we are given a directed graph $G=(V,E)$, where $V$ is the set of vertices, and $E$ is the set of edges between vertices. We typically define $n$ to be the number of vertices $\abs{V}$ and $m$ to be the number of edges $\abs{E}$. Assume two disjoint sets of vertices $A$ and $B$. We will write $E(A, B) = \set{e = (v, w) \in E \mid v \in A, w \in B}$ to be the edges from $A$ to $B$.

We assume that we are given a function $\funcdef{H}{2^V}{\real}$ that measures the incoherency of a vertex set.
We solely analyze $L_2$ loss as a measure of incoherence, but the same problem formulation can be interesting for other functions depending on the types of features available. More specifically, assume that we have a map $\funcdef{a}{V}{\real^D}$ that maps a vertex $v$ to a real-valued vector of $D$ features $a(v)$. Then the measure is $L_2(S) = \min_{\mu} \sum_{v \in S} \norm{a(v) - \mu}_2^2$, where $\mu$ is the centroid $\mu = \frac{1}{|S|} \sum_{v \in S} a(v)$.

Our goal is to partition the graph into a sequence of $k$ groups that are simultaneously coherent and minimize cross-edges. To measure the cost of such a partition, we introduce two weight parameters $\lambda_f$ and $\lambda_b$ for the forward and backward edges, respectively. Given an ordered partition $\fm{S} = S_1, \ldots, S_k$, we define a cost function $q$ as
\[
    q(\fm{S} \mid \lambda_f, \lambda_b, H) = \sum_{i = 1}^k H(S_i) + \sum_{j = i+1}^k \lambda_f \abs{E(S_i, S_j)} +  \lambda_b \abs{E(S_j, S_i)}\quad.
\]
We often drop $H$, $\lambda_f$, or $\lambda_b$ from the notation if they are clear from the context. For a set $S_i$ in the partition, we use the terms group and cluster interchangeably.

The definition of $q$ immediately leads to our main optimization problem

\begin{problem}[directed graph segmentation (\prbgts)]
Given a directed graph $G = (V, E)$, integer $k$, two weights $\lambda_f$ and $\lambda_b$,
and a function $\funcdef{H}{2^V}{\real}$, find an ordered $k$-partition $\fm{S} = S_1, \ldots, S_k$ of $V$ such that $q(\fm{S} \mid \lambda_f, \lambda_b, H)$ is minimized.
\end{problem}

From now on we will focus on $H = L_2$ case.
Note that if we set $\lambda_b = \lambda_f$, then the order of sets in $\fm{S}$ does not matter. Moreover, if $\lambda_b = \lambda_f = 0$, then \prbgts reduces to $k$-means clustering. However, our methods can be used with other centroid-based losses.

\section{Related work}\label{sec:related}
Clustering is a staple method in supervised learning with $k$-means problem (see~\citep{duda2000pattern}, for example) being the most common optimization problem. The \np-hardness of clustering, even in the plane~\citep{mahajan2012planar}, makes our problem immediately \np-hard when we set $\lambda_f = \lambda_b = 0$ and $H$ to be $L_2$ loss.

\textbf{Constrained clustering}.
A framework similar to our problem setting is pairwise constrained clustering (PCC), where selected pairs of data points must be in the same cluster or must belong to different clusters~\citep{davidson2005clustering,wagstaff2001constrained}.
Other constraints, such as balancing constraints or minimum-size constraints, have also been studied; we refer the reader to~\citep{basu2008constrained} for more details. The key technical difference is that in PCC, the constraints have no direction. Consequently, the order of the resulting clusters does not matter. However, in our case, if $\lambda_f \neq \lambda_b$ the order of groups matters, especially if we set $\lambda_b = \infty$ and $\lambda_f = 0$.

\textbf{Network clustering}.
We can also view our problem as a directed network clustering problem.
Undirected graph clustering has been well-studied. Popular methods include minimizing modularity~\citep{newman2004finding} as well as stochastic blockmodelling~\citep{abbe2017community}, spectral clustering~\citep{von2007tutorial}, or closely related normalized cuts~\citep{meilua2007clustering}.
We refer the reader to~\citep{fortunato2010community,schaeffer2007graph} for surveys on undirected graph clustering.

The clustering of directed graphs poses additional challenges, as measures need to be adapted. \citet{leicht2008community} proposed a modularity measure for directed graphs. \citet{chung2005laplacians} proposed a Laplacian matrix for directed graphs, allowing the use of spectral clustering. Moreover, a random-walk approach was proposed by~\citet{rosvall2008maps}. We refer the reader to~\citep{malliaros2013clustering} for a survey on the clustering of directed graphs.
The main difference between graph clustering and our problem is that graph clustering methods focus on optimizing measures based solely on edges. In contrast, we use additional features with $L_2$ loss while also minimizing the number of cross edges.

\textbf{Partial graph ordering}.
Our problem is closely related to finding a total preordering (a weak order) with $k$ buckets of the vertices in a directed graph. In the \textsc{Max-$k$-Ordering}\xspace problem, the aim is to assign the vertices of a directed graph to $k$ groups such that the number of forward edges between groups is maximized. \citet{kenkre2017approximability} provide a tight 2-approximation using a linear programming relaxation. This corresponds to a version of our problem where $L_2$ losses are always $0$, $\lambda_b=0$, and $\lambda_f < 0$. 

If instead $\lambda_f = -\lambda_b$ with no $L_2$ terms, we can normalize the costs for forward and backward edges to be $-\frac{1}{2}$ and $-\frac{1}{2}$, respectively. If we then define the number of groups $k$ as a free parameter, the problem can be transformed to a bucket ordering problem discussed by \citet{gionis2006algorithms}, where the aim is to place items in buckets such that the ordering of the buckets maximally agrees with a given pair order matrix.

The key difference between prior graph ordering work and ours is that our problem includes $L_2$ loss, which discourages placing all nodes in the same cluster. In contrast, previous work does this by maximizing the number of forward edges between the groups, while in our case, having all edges within the groups is preferred as $\lambda_f > 0$.

\textbf{Segmentation}.
An interesting special case of our problem occurs when the underlying graph is a directed path, and we set the backward weight to $\lambda_b = \infty$. In this case, the clusters will respect the order of the vertices, and \prbgts reduces to a segmentation problem, in which we are given a \emph{sequence} of points and are asked to segment the sequence into $k$ coherent groups. Segmentation can be solved with dynamic programming in quadratic time~\citep{bellman:61:on} and can be efficiently approximated in quasilinear time~\citep{guha:01:estimate} or linear time~\citep{guha:06:estimate,tatti2019segmentation}.

\textbf{Isotonic regression}.
Finally, let us point out an interesting connection to isotonic regression~\citep{kyng2015fast}.
Assume that the underlying graph is a DAG.
If we set
$\lambda_b = \infty$ and $\lambda_f = 0$, use $L_2$ loss, and additionally require that the $L_2$ norms of the centroids need to be monotonically increasing $\lVert \mu_{i - 1} \rVert_2 < \lVert \mu_i \rVert_2$. Then we can show that the optimization problem can be solved in polynomial time by first applying isotonic regression, ordering the nodes by the obtained mapping, and segmenting the nodes in $k$ segments using dynamic programming~\citep{bellman:61:on}.
\section{Mixed-integer linear program solving directed graph segmentation}
\label{sec:milp}

In this section, we formulate \prbgts as a mixed-integer linear program (MILP) by modifying the MILP formulation for the $k$-means clustering problem by \citet{agoston2024mixed} and introducing additional terms representing the costs for forward and backward edges. We can use this solver directly to solve the segmentation
for the smaller networks. We will also use this program as a base for linear programs that we will introduce later.

Using a well-known identity, we can write the $L_2$ error term as 
\begin{equation*}
\sum_{i = 1}^k \sum_{v \in S_i} \norm{a(v) - \mu_i}_2^2 = \sum_{i = 1}^k \frac{1}{2|S_i|}\sum_{u, v \in S_i} (a(u)-a(v))^2.
\end{equation*}

Our objective is then to minimize
\begin{equation}
\label{eq:main_lp_objective}
\frac{1}{2}\sum_{u, v \in V} (a(u)-a(v))^2 z_{uv} + \sum_{(u,v) \in E} \left( f_{uv}\lambda_f + b_{uv}\lambda_b \right)
\end{equation}
over the variables $z_{uv}$, and $f_{uv}$ and $b_{uv}$ subject to constraints that we will introduce next.

Here, the variables $z_{uv}$ are fractional variables that are $0$ when vertices $u$ and $v$ are in different groups, and when $u$ and $v$ are in the same group, $z_{uv} = 1/C$, where $C$ is the size of that group. Moreover, we will define $b_{uv}$ so that $b_{uv} = 1$ when the edge $(u, v) \in E$ is a backward edge, which means $u$ and $v$ are in different groups $u \in S_i$ and $v \in S_j$ such that $j < i$. Similarly, we will define $f$ so that $f_{uv} = 1$ whenever $(u,v)$ is a forward edge, and $f_{uv} = 0$ otherwise. 

We will divide the constraints into several groups and define several auxiliary variables.
    
1. We start by constraining $z$, 
\begin{subequations}
\label{eq:lp_start}
\begin{align}
\sum_{v \in V} z_{uv} &= 1 & u \in V \\
z_{uv} &\leq z_{uu} & u, v \in V\\
\sum_{u \in V} z_{uu} &= k \\
z_{uv} &\geq 0 & u,v \in V.
\end{align}
\end{subequations}

2. Next we define variables $x_u^i$ to indicate whether the vertex $u$ is assigned to $S_i$,
\begin{subequations}
\label{eq:lp_x_start}
\begin{align}
\sum_{i=1}^k x_u^i &= 1, & u \in V\\
\sum_{u \in V} x_u^i &\geq 1, & i \in \{1,\ldots,k\} \label{eq:lp_x_ne} \\
x_u^i &\in \{0, 1\}. 
\end{align}
\end{subequations}

3. If a vertex $v$ is in a different group $i$ than vertex $u$, then $z_{uv}$ is to be zero,
\begin{align}
z_{uv} & \leq 1 + x_{u}^i - x_{v}^i & u, v \in V, i \in \{1,\ldots,k\}.
\end{align}

4. Our next goal is to constrain $b$ and $f$.
Note that by definition $b_{uv}$ is equivalent to 
\[
    b_{uv} = \sum_{i=1}^k \min(x_v^i, \max(p_v^i - p_u^i, 0)),
\] where the variables $p_{u}^i$ indicate that vertex $u$ belongs to a possibly earlier group $S_j$ with $j \leq i$ (we will define these variables later).
For $f_{uv}$, we similarly have
\[
    f_{uv} = \sum_{i=1}^k \min(x_v^i, \max(q_v^i - q_u^i, 0)),
\] where the variables $q_{u}^i$ indicate that vertex $u$ belongs a possibly later group $S_j$ with $j \geq i$. 

In order to constraint $f$ and $b$, let us now define the variables $p_v^i$ and
an auxiliary variables $\alpha_{uv}^i$ that corresponds to $\max(p_v^i - p_u^i, 0)$,

\begin{subequations}
\label{eq:lp_p_start}
\begin{align}
p_u^i &= \sum_{j=1}^i x_u^j, & u \in V, i \in \{0,\ldots, k\} \label{eq:lp_p}\\
\alpha_{uv}^i &\geq p_v^i - p_u^i, & (u, v) \in E, i \in \{1,\ldots, k\}\\
\alpha_{uv}^i &\geq 0, & (u, v) \in E, i \in \{1,\ldots, k\}.
\end{align} 

We similarly define $q_v^i$ and $\gamma_{uv}^i = \max(q_v^i - q_u^i, 0)$ for the forward edges.

\begin{align}
q_u^i & 
= 1 - p_u^{i-1}, & u \in V, i \in \{1,\ldots, k\}  \label{eq:lp_q}\\
\gamma_{uv}^i &\geq q_v^i - q_u^i, & (u, v) \in E, i \in \{1,\ldots, k\}\\
\gamma_{uv}^i &\geq 0, & (u, v) \in E, i \in \{1,\ldots, k\},
\end{align}
\end{subequations}

5.
Note that $b_{uv} = \sum_i \min(x^i_v, \alpha_{uv}^i)$. To express this, we can introduce additional binary variable sets $s_{uv}^i$ that indicate which of the two terms is smaller and variables $\beta_{uv}^i$ to store the minimum,
\begin{subequations}
\label{eq:lp_d_start}
\begin{align}
x_v^i &\leq \alpha_{uv}^i+s_{uv}^i, & (u, v) \in E, i \in \{1,\ldots, k\}\\
\alpha_{uv}^i &\leq x_v^i + 1 - s_{uv}^i, & (u, v) \in E, i \in \{1,\ldots, k\}\\
\beta_{uv}^i &\geq x_v^i - s_{uv}^i, & (u, v) \in E, i \in \{1,\ldots, k\}\\
\beta_{uv}^i &\geq \alpha_{uv}^i + s_{uv}^i - 1, & (u, v) \in E, i \in \{1,\ldots, k\}\\
s_{uv}^i &\in \{0, 1\}\\ 
b_{uv} &= \sum_{i=1}^k \beta_{uv}^i, & (u, v) \in E.
\end{align}

Similarly, for the forward edges, we define corresponding variables $t_{uv}$ and $\delta_{uv}^i$,
\begin{align}
x_v^i &\leq \gamma_{uv}^i+t_{uv}^i, & (u, v) \in E, i \in \{1,\ldots, k\}\\
\gamma_{uv}^i &\leq x_v^i + 1 - t_{uv}^i, & (u, v) \in E, i \in \{1,\ldots, k\}\\
\delta_{uv}^i &\geq x_v^i - t_{uv}^i, & (u, v) \in E, i \in \{1,\ldots, k\}\\
\delta_{uv}^i &\geq \gamma_{uv}^i + t_{uv}^i - 1, & (u, v) \in E, i \in \{1,\ldots, k\}\\
t_{uv}^i &\in \{0, 1\}\\ 
f_{uv} &= \sum_{i=1}^k \delta_{uv}^i, & (u, v) \in E.
\end{align}
\end{subequations}

The above discussion leads to the following result.

\begin{theorem}
Solving Eq.~\ref{eq:main_lp_objective} subject to Eqs.~\ref{eq:lp_start}--\ref{eq:lp_d_start} solves \prbgts.
The number of variables and constraints in MILP is in $\bigO{n^2 + k(m + n)}$.
\end{theorem}
\section{Iterative approach}\label{sec:iterative}

A standard algorithm to solve the $k$-means problem is to iteratively fix centroids and optimize the partition, and then optimize centroids while keeping the partition fixed.

In order to adapt this idea to our approach, assume that we are given an ordered $k$-partition $\fm{S}$ of vertices $V$ and $k$ centroids. 
Let us define the cost
\[
    q(\fm{S}, \set{\mu_i} \mid \lambda_f, \lambda_b) = \sum_{i = 1}^k \sum_{v \in S_i} \norm{a(v) - \mu_i}_2^2 + \sum_{j = i+1}^k \lambda_f \abs{E(S_i, S_j)} +  \lambda_b \abs{E(S_j, S_i)}\quad.
\]

We will then consider two related sub-problems: in the first, we optimize the partition while keeping the centroid fixed, while in the second we, optimize the centroids while keeping the partition fixed.

\begin{problem}[\prbgtspalt]
Given a directed graph $G = (V, E)$, integer $k$, two weights $\lambda_f$ and $\lambda_b$,
and $k$ centroids $\mu_1, \ldots, \mu_k$, find an ordered $k$-partition $\fm{S} = S_1, \ldots, S_k$ of $V$ such that $q(\fm{S}, \set{\mu_i} \mid \lambda_f, \lambda_b)$ is minimized.
\end{problem}

\begin{problem}[\prbgtsc]
Given a directed graph $G = (V, E)$, integer $k$, two weights $\lambda_f$ and $\lambda_b$, and an ordered $k$-partition $\fm{S} = S_1, \ldots, S_k$ of $V$, find
$k$ centroids $\mu_1, \ldots, \mu_k$ such that $q(\fm{S}, \set{\mu_i} \mid \lambda_f, \lambda_b)$ is minimized.
\end{problem}

However, unlike in standard $k$-means, the order of the clusters is also important when $\lambda_f \neq \lambda_b$. In preliminary experiments, we found that an iterative approach may get stuck with a bad order of centroids. To avoid this, we define a third subproblem where we sort a fixed set of clusters to minimize the number of backward edges between clusters when $\lambda_b > \lambda_f$, or forward edges when $\lambda_f > \lambda_b$.

\begin{problem}[\prbgtss]
Given a directed graph $G = (V, E)$, integer $k$, two weights $\lambda_f$ and $\lambda_b$, and a $k$-partition $\fm{S} = S_1, \ldots, S_k$ of $V$ and the centroids
$\mu_1, \ldots, \mu_k$ sort $\fm{S}$ such that $q(\fm{S}, \set{\mu_i} \mid \lambda_f, \lambda_b)$ is minimized.
\end{problem}

By iteratively solving these three subproblems, we get a heuristic algorithm for \prbgts. The pseudo-code for this approach is given in Algorithm~\ref{alg:iter}.

\begin{algorithm}[t]
\caption{Iterative algorithm}
\label{alg:iter}
$\fm{S} = S_1, \ldots, S_k \gets$ initial partition\;
\While{the loss is decreasing or until a set amount of iterations} {
    compute the centroids $\mu_i = \frac{1}{|S_i|} \sum_{v \in S_i} a(v)$ for $i = 1, \ldots, k$\;
    sort clusters to minimize backward (forward) edges if $\lambda_b > \lambda_f$ ($\lambda_f > \lambda_b$)\;
    optimize $\fm{S}$ minimizing $q(\fm{S}, M \mid \lambda_f, \lambda_b)$ while the centroids $M=\{\mu_i \mid i = 1,\ldots,k\}$ remain fixed\;
}
\Return sets $S_1, \ldots, S_k$\;
\end{algorithm}

Note that it is possible for a solution to \prbgtspalt to contain empty sets, effectively partitioning the vertices into fewer than $k$ groups. To ensure a partition into exactly $k$ sets, we consider a restricted version of the problem, requiring the sets to be non-empty.

\begin{problem}[\prbgtsp]
Given a directed graph $G = (V, E)$, integer $k$, two weights $\lambda_f$ and $\lambda_b$,
and $k$ centroids $\mu_1, \ldots, \mu_k$, find an ordered $k$-partition $\fm{S} = S_1, \ldots, S_k$ of $V$ with no empty sets such that $q(\fm{S}, \set{\mu_i} \mid \lambda_f, \lambda_b)$ is minimized.
\end{problem}

However, as we shall see later in Theorem~\ref{thr:apx_hard} and in Section~\ref{sec:sub}, \prbgtsp is computationally challenging and often results in slower versions of our algorithms. In practice, we focus on solving \prbgtspalt and populate empty clusters between iterations as needed using a greedy heuristic that finds an individual best vertex to move to each empty cluster.

\section{Computational complexity}\label{sec:complexity}
Note that \prbgtsc has an analytical solution, $\mu_i = \frac{1}{\abs{S_i}} \sum_{v \in S_i} a(v)$. However, the partition problems are \np-hard.

\begin{theorem}
\prbgtspalt and \prbgtsp are \np-hard problems even if we fix $\lambda_f, \lambda_b$, and $k$ as long as $\lambda_f + \lambda_b > 0$ and $k \geq 3$.
\end{theorem}

\begin{proof}

The following proof will work directly for both problems. For simplicity, let us focus only on \prbgtspalt.

We will show that the unweighted minimum multiterminal cut problem\footnote{This problem is also known as the multiway cut problem.} (\prbmtc) can be reduced to \prbgtspalt. \prbmtc is
an \np-hard problem~\citep{dahlhaus1994complexity} where we are given an undirected graph and a set of terminals
$T = t_1,\ldots,t_k$, and are asked to partition the vertices in $k$ groups $\fm{C} = C_1, \ldots, C_k$ such that $t_i \in C_i$ and the number of cross-edges is minimized.

Let $G=(V, E)$ be an undirected instance of \prbmtc with $k$ terminals $T = t_1,\ldots,t_k$.
Create an instance of \prbgtspalt as follows:
Set the number of disjoint sets to find as $k$.
Define $k$ centroids $\mu_{1},\ldots,\mu_{k}$ to be orthogonal vectors of length $k$ such that $\mu_i$ has $\lambda_f+\lambda_b$ as the $i$th entry and 0 as all other entries.
Create graph $G'$ as an instance of \prbgtspalt so that it contains all the vertices in $V$ and each undirected edge in $E$ becomes two directed edges, one in each direction.
Additionally, for each terminal $t_i$ create a set $U_i$ of $\abs{U_i} = \abs{V}$ new vertices that are only connected to $t_i$, each with two edges of opposite direction. We set the feature vectors $a(t_i) = \mu_i$ and $a(u) = \mu_i$ for any $u \in U_i$. The remaining vertices $v \in V \setminus T$ have $a(v) = 0$.

Let $S_1, \ldots, S_k$ be the solution for \prbgtspalt.

The cost of including a vertex $u \in U_i$ in $S_i$ is
$\norm{a(u) - \mu_i}^2 = 0$, while the cost of including $u$ in $S_{j}$, for $j \neq i$, is
$\norm{a(u) - \mu_j}^2 = 2(\lambda_f+\lambda_b)$. It is then optimal to include $u \in U_i$ in $S_i$ as the possible loss of $\lambda_f+\lambda_b$ from the edges between the vertex $u$ and $t_i$ is less than the loss of $2(\lambda_f+\lambda_b)$ that we would have if $u$ were in another set. Therefore, an optimal solution will include all the vertices in $U_i$ in $S_i$, which also means that $\set{S_i}$ solves \prbgtsp. Moreover, this means that the terminal $t_i$ will also have to be in $S_i$ as otherwise, the cost from the edges between $t_i$ and the $|U_i| = |V|$ vertices in $U_i$ will be more than any possible loss from the edges between $t_i$ and any other vertices $v \in V$. Thus, $t_i \in S_i$.

Finally, the remaining nonterminal vertices in $V \setminus T$ have the same loss of $\lambda_f+\lambda_b$ due to features, regardless of which set they belong to, so an optimal solution will assign them such that the cost arising from the edges between the sets is minimized.

Given a partition $\fm{S}$, we define a cut $\fm{C}$ for $G$ by setting $C_i = V \cap S_i$. This is a valid cut since $t_i \in S_i$, and since \prbgtspalt minimizes the number of cross-edges, $\fm{C}$ is optimal. 
\end{proof}

\begin{theorem}
\label{thr:apx_hard}
\prbgtsp and \prbgts are \apx-hard problems even if we fix $\lambda_f = 0$, $\lambda_b > 0$, and $k = n$ where $n$ is the number of nodes in the graph.
If the unique game conjecture is true, then there is no constant approximation algorithm for \prbgtsp or \prbgts. If ETH is true, then there is no approximation algorithm for \prbgtsp or \prbgts with $7/6 - \epsilon$ guarantee that runs in $\bigOstar{2^{n^{1 - \epsilon}}}$.
\end{theorem}

\begin{proof}
We will prove the claim with a reduction from Feedback Arc Set (\prbfas). In \prbfas we are given a directed graph and the goal is to order the nodes while minimizing the number of backward edges. Assume that we are given such a graph $G = (V, E)$ with $n$ nodes. We will set no features for the nodes, and set $k = n$. We will also set $\lambda_f = 0$ and $\lambda_b$ to any positive number. Since $k = n$, every node must be in its own cluster when solving \prbgtsp (or \prbgts), inducing an order among the nodes. Moreover, the objective function is equal to the number of backward edges times $\lambda_b$. The \apx-hardness~\citep{kann1992approximability} of \prbfas now proves the result. The remaining claims follow from inapproximability results of \prbfas by~\citet[Corollary 1.2]{guruswami2011beating}, and~\citet[Theorem 5]{bonnet2018sparsification}.
\end{proof}

Next, we show that also \prbgtss is a hard problem.

\begin{theorem}
\prbgtss is an \apx-hard problem even if we fix $\lambda_f = 0$, $\lambda_b > 0$.
If the unique game conjecture is true, then there is no constant approximation algorithm for \prbgtss. If ETH is true, then there is no approximation algorithm for \prbgtss $7/6 - \epsilon$ guarantee that runs in $\bigOstar{2^{n^{1 - \epsilon}}}$.
\end{theorem}

\begin{proof}
By setting $\lambda_f = 0$ and $\lambda_b > 0$, \prbgtss is equivalent to \prbfas. The results follow the hardness results by~\citet{kann1992approximability}, by~\citet[Corollary 1.2]{guruswami2011beating}, and~\citet[Theorem 5]{bonnet2018sparsification}.
\end{proof}

We should point out that despite this result, we are often able to solve \prbgtss exactly. This is done
by contracting each cluster to a single node, and then applying, for example, a weighted \prbfas solver based on integer programming.\!\footnote{\url{https://python.igraph.org/en/stable/}} Since the contracted graph contains only $k$ nodes, running an exponential algorithm is feasible for small values of $k$. For larger values, we resort to a heuristic by~\citet{eades1993fast}.
\section{Special cases for finding the optimal partition exactly in polynomial time}\label{sec:sub}

We showed that solving \prbgtspalt in general is \np-hard.
However, there are two cases where we can solve \prbgtspalt in polynomial time: ($i$) the input graph is a tree or
($ii$) $k = 2$. In this section, we will consider these two cases. The more general case is discussed in the next section.

We should point out that while we focus on $L_2$ loss, this approach works with any loss function as long as it can be decomposed as a sum over the nodes.

\textbf{Case when the input graph is a tree.}
We will first consider a case when the input graph $G$ is a tree.
For simplicity, we will assume that $G$ is also arborescence,
that is, there is a root vertex, say $r$, from which there is a path to each vertex with a \emph{directed} path, but we can extend this approach to trees and forests. 

Given an arborescence $G = (V, E)$ and a vertex $v \in V$, let $G_v$ be the subtree containing $v$ and its descendants. We define $c[v, i]$ to be the cost of the optimal partition $\fm{S}$ of $G_v$ such that $v \in S_i$.
Note that $\min_i c[r, i]$ is equal to the cost of the solution to \prbgtspalt.

In order to compute $c[v, i]$, let us first define 
\[
    \ell[v, i] = \min_{j \leq i} c[v, j] \quad\text{and}\quad
    u[v, i] = \min_{j \geq i} c[v, j],
\]
that is, $\ell[v, i]$ is the cost of the optimal partition $\fm{S}$ of $G_v$ such that $v \in S_j$ for some $j \leq i$. Similarly, $u[v, i]$ is the cost of the optimal partition $\fm{S}$ of $G_v$ such that $v \in S_j$ for some $j \geq i$.
For simplicity, we define $\ell[v, 0] = u[v, k+1] = \infty$.

Next, we compute $c[v, i]$ using only $u$ and $\ell$ of the children of $v$.

\begin{theorem}
Let $c$, $u$, and $v$ be as above. Then for $v \in V$ and $i \in 1, \ldots, k$, 
\begin{equation}
\label{eq:dp}
    c[v, i] = \norm{a(v) - \mu_i}^2 +
    \sum_{w \mid (v, w) \in E} \min(c[w, i], \lambda_b + \ell[w, i-1], \lambda_f + u[w, i+1]).
\end{equation}
\end{theorem}

\begin{proof}
Define for notational simplicity $M = \mu_1, \ldots, \mu_k$.
Let $\fm{S}$ be the partition responsible for $c[v, i]$.
For any child $w$ of $v$, let us write $\fm{S}_w$ to be $\fm{S}$ projected
to $G_w$. Let $g(w)$ be the (possibly zero) cost of the possible cross edge $(v, w)$. Since $G_v$ is a tree, we can decompose the cost as
\[
    q(\fm{S}, M) = \norm{a(v) - \mu_i}^2 + \sum_{w \mid (v, w) \in E} g(w) + q(\fm{S}_w, M).
\]
Let $w$ be a child of $v$. We have 3 possible cases. If $w \in S_i$, then $g(w) = 0$ and, due to optimality, $q(\fm{S}_w, M) = c[w, i]$. If  $w \in S_j$ for $j < i$, then
$g(w) = \lambda_b$ and, due to optimality, $q(\fm{S}_w, M) = \ell[w, i-1]$.
Similarly, if  $w \in S_j$ for $j > i$, then
$g(w) = \lambda_f$ and, due to optimality, $q(\fm{S}_w, M) = u[w, i+1]$.
Finally, due to the optimality, the actual case will be the one yielding the smallest cost.
\end{proof}

Computing $c[v, i]$ requires $\bigO{d + \deg v}$ time, where $d$ is the length of the feature vector. Since $\ell[v, i] = \min(\ell[v, i - 1], c[v, i])$  and $u[v, i] = \min(u[v, i + 1], c[v, i])$ we can compute both quantities in constant time. In summary, we can find the optimal cost in $\bigO{knd + km} = \bigO{knd}$ time, since $m = n-1$ in a tree.
To obtain the corresponding partition we store the indices that were responsible for $c[v, i]$ in Eq.~\ref{eq:dp}.

In summary, if the input graph is a tree, we can search for the partition using Algorithm~\ref{alg:iter} and solve the sub-problem \prbgtspalt with dynamic programming. 

\textbf{Requiring non-empty clusters.}
Let us now consider \prbgtsp, that is we require that the clusters should always be non-empty.
In this case we cannot use the same dynamic program. However, a similar approach leads to a program that runs in $\bigO{knd + kn4^k}$ which may be feasible for small numbers of $k$. We leave the case whether \prbgtsp is \np-hard when the input graph is a tree as an open problem.

Similarly to the first case, given an \emph{ordered} arborescence $G = (V, E)$, a vertex $v \in V$ and an integer $s \in [k]$, let $G_{v, s}$ be the subtree containing $v$ and $s$ first children of $v$, and their descendants. Assume a set of indices $Z \subseteq [k]$. We define $c[v, s, i, Z]$ to be the cost of the optimal partition $\fm{S}$ of $G_{v, j}$ such that $v \in S_i$ and the index set of non-empty sets in $\fm{S}$ is a superset of $Z$.
Note that $\min_i c[r \deg(r), i, [k]]$ is equal to the cost of the solution to \prbgtsp.

To compute the dynamic program, we first define the helper tables,
\[
    \ell[v, i, Z] = \min_{j \leq i} c[v, \deg(v), j, Z] \quad\text{and}\quad
    u[v, j, Z] = \min_{j \geq i} c[v, \deg(v), j, Z],
\]

Given a vertex $v$ and its $s$th child $w$, we can compute the entry in the table using the identity,
\[
    c[v, s, i, Z] = \norm{a(v) - \mu_i}^2 +
    \min_{X, Y \mid X \cup Y = Z \setminus \set{i}}
    c[v, s - 1, i, X] + f[w, i, Y],
\]
where
\[
    f[w, i, Y] = \min(c[w, i, Y], \lambda_b + \ell[w, i - 1, Y], \lambda_f + u[w, i + 1, Y]).
\]

Computing $c[v, s, i, Z]$ requires $\bigO{d + 2^k}$ time, where $d$ is the length of the feature vector. Note that $\norm{a(v) - \mu_i}^2$ does not depend on $s$ or $Z$.
In addition, the number of valid $(v, s)$ pairs is in $\bigO{n}$ since $G$ is a tree. Consequently, the running time required to compute $c$ is in $\bigO{knd + kn4^k}$.

\textbf{Case when $k = 2$.} 
Next, we allow the input graph to be any directed graph but we require that $k = 2$. We will argue that we can then solve \prbgtspalt using a weighted minimum directed $s$-$t$ cut. The same approach can be also used to solve \prbgtsp.

To do the mapping assume that we are given a graph $G = (V, E)$ and two centroids $\mu_1$ and $\mu_2$. We define a weighted directed graph $H = (W, A)$
as follows. The vertices $W$ consist of $V$ and two additional vertices $s$ and $t$. For each $v$ we introduce an edge $(s, v)$ to $A$ with a weight $c(s, v) = \norm{v - \mu_2}^2$ and an edge $(v, t)$ with a weight $c(v, t) = \norm{v - \mu_1}^2$. 
For each $(v, w) \in E$, we add an edge $(v, w)$ with a weight $\lambda_f$ and an edge $(w, v)$ with a weight $\lambda_b$.

The next theorem connects the $s$-$t$ cut with the cost of the partition.

\begin{theorem}
\label{thr:mincut}
Let $C_1, C_2$ be the $s$-$t$ cut for $H$. Let $\fm{S} = S_1,S_2$ where $S_i = C_i \cap V$.
Then $q(\fm{S}, \mu_1, \mu_2)$ is equal to the total weight of edges from $C_1$ to $C_2$.
\end{theorem}

\begin{proof}
The cost of the partition is equal to
\[
\begin{split}
    q(\fm{S}, \mu_1, \mu_2) = & \sum_{v \in S_1} \norm{v - \mu_1}^2 + \sum_{v \in S_2} \norm{v - \mu_2}^2
    + \sum_{e \in E(S_1, S_2)} \lambda_f
    + \sum_{e \in E(S_2, S_1)} \lambda_b \\
    = & \sum_{v \in S_1} c(v, t) + \sum_{v \in S_2} c(s, v)
    + \sum_{e \in A(S_1, S_2)} c(e).
\end{split}
\]
The sums amount to the total weight of edges from $C_1$ to $C_2$.
\end{proof}

The theorem states that we can solve \prbgtspalt with a minimum cut on $H$.
Finding the minimum cut can be done in $\bigO{nm}$ time~\citep{orlin2013max}, though theoretically slower algorithms, e.g., by~\citet{boykov2004experimental}, are faster in practice. Moreover, there are several theoretically faster algorithms,  e.g., by~\citet{bernstein2024maximum}, but they are randomized, and the probability of failure must be taken into account if they were to be used. Constructing the graph requires $\bigO{nd + m}$ time, where $d$ is the length of the feature vectors. Thus, we can solve \prbgtspalt for $k = 2$ in $\bigO{n(d + m)}$ time.

We can also use a similar approach to solve \prbgtsp. We can do this by selecting two nodes, say $v_1$ and $v_2$ and forcing $v_1 \in S_1$ and $v_2 \in S_2$ by setting the weights of $(s, v_1)$ and $(v_2, t)$ to infinity. In order to solve \prbgtsp we need to test every possible pair for $v_1, v_2$, which leads to a running time of $\bigO{n^3(d + m)}$.
\section{Linear programming approaches for finding an approximately optimal partition}\label{sec:lp}

In this section, we show that by using a randomized rounding algorithm for a linear programming relaxation, we get a $k-1$ approximation for \prbgtspalt. we also show that when $\lambda_f = \lambda_b$, we can improve the approximation factor to $\frac{k+1}{3}$.

Recall the MILP formulation for \prbgts given in Section~\ref{sec:milp}.
We can also express the \prbgtsp as a simpler ILP.
To do that, we define $c_u^i$ to be
the cost of assigning a vertex $u$ to a cluster (centroid) $i$. We can then formulate \prbgtsp as an ILP by removing $z_{uv}$ from the previous program and using $x_v^i$ directly.
This leads to the following ILP:

\begin{align}
\label{eq:lp_objective}
\textsc{Minimize} \quad & \sum_{u \in V} \left( \sum_{i=1}^k x_{u}^i c_{u}^i \right) + 
\sum_{(u,v) \in E} \left( f_{uv}\lambda_f + b_{uv}\lambda_b \right) \\
\textsc{subject to}  \quad & \text{Constraints \ref{eq:lp_x_start} and \ref{eq:lp_p_start}--\ref{eq:lp_d_start}}. \nonumber
\end{align}

Note that Constraint~\ref{eq:lp_x_ne} forces every cluster to be non-empty. Consequently, dropping it will yield an ILP that solves \prbgtspalt.

\subsection{Relaxation and rounding algorithm}

In this section, we introduce an approach based on a linear program for \prbgtspalt.

A common approximation technique is to transform the ILP into a linear program, solve the program, and apply a rounding algorithm. Unfortunately, the variables in Constraints~\ref{eq:lp_d_start} are problematic as they require additional binary variables $s_{uv}^i$ and $t_{uv}^i$.

Instead, we consider an alternative program, where we replace the variables $b$ and $f$ with
\begin{align*}
b'_{uv} &= \sum_{i=1}^k \max(p_v^i - p_u^i, 0), &
f'_{uv} &= \sum_{i=1}^k \max(q_v^i - q_u^i, 0).
\end{align*}

We can express $b'$ and $f'$ without the additional binary variables $s_{uv}^i$ and $t_{uv}^i$, and Constraints~\ref{eq:lp_d_start},
\begin{align}
\label{eq:b'_def}
b'_{uv} &=  \sum_{i=1}^k \alpha_{uv}^i, & (u, v) \in E, \\
\label{eq:f'_def}
f'_{uv} &= \sum_{i=1}^k \gamma_{uv}^i, & (u, v) \in E.
\end{align}
Replacing $b$ and $f$ with $b'$ and $f'$ in the objective function leads to the following integer linear program.

\begin{align}
\label{eq:lp'_objective}
\textsc{Minimize} \quad &
\sum_{u \in V} \left( \sum_{i=1}^k x_{u}^i c_{u}^i \right) +
\sum_{(u,v) \in E} \left( f_{uv}'\lambda_f + b'_{uv}\lambda_b \right) \\
\textsc{subject to} \quad & \text{Constraints~\ref{eq:lp_x_start} (except~\ref{eq:lp_x_ne}), \ref{eq:lp_p_start}, and~\ref{eq:b'_def}--\ref{eq:f'_def}}. \nonumber
\end{align}

Next, consider the linear programming relaxation, allowing the variables $x_u^i$ to take on real values $x_u^i \geq 0$.

\begin{algorithm}[t!]
\caption{LP rounding algorithm}
\label{alg:rounding}
Solve the LP relaxation for Equation~\ref{eq:lp'_objective} to obtain variables $x_v^i$ \;
Randomly sample a variable $\rho$ from the interval $[0, 1]$\;
\lFor{$v \in V$}{
    assign $v$ to the set $S_i$ where $i$ satisfies $\sum_{j=1}^{i-1} x_v^j < \rho \leq \sum_{j=1}^i x_v^j$
}
\Return sets $S_1, \ldots, S_k$\;
\end{algorithm}

\begin{theorem}
\label{thr:lp_approx}
The rounding algorithm in Algorithm~\ref{alg:rounding} provides a $k-1$-approximation for \prbgtspalt in expectation.
\end{theorem}

To prove the theorem, we first need the following lemma.

\begin{lemma}
\label{lem:relax}
Let $\opt$ denote the cost of the solution to \prbgtspalt. Let $\opt'$ be the cost of an optimal ILP solution to the objective in Eq.~\ref{eq:lp'_objective}. Then $\opt' \leq (k - 1)\opt$.
\end{lemma}

\begin{proof}
Let $x_v^i$, $b_{uv}$, $f_{uv}$, $p^i_v$, $q^i_u$, $\alpha^i_{uv}$, and $\gamma^i_{uv}$ be the variables yielding $\opt$.
Define $b'_{uv}$, $f'_{uv}$ using Eqs.~\ref{eq:b'_def}--\ref{eq:f'_def}.

Note that $b'_{uv}$ is the difference in the group indices $\max(j-i, 0)$ when $v$ belongs to an earlier group $v \in S_i$ than $u \in S_j$. Therefore, it is at most $k-1$ when $b_{uv} = 1$, and $0$ otherwise. Thus, we have 
\[
    b'_{uv} \leq (k-1) b_{uv}.
\]
Similarly, for the forward edges, we have 
\[
    f'_{uv} \leq (k-1) f_{uv}.
\]

The variables constitute a feasible solution for the ILP given in Eq.~\ref{eq:lp'_objective}. Immediately,
    $\opt' \leq (k-1) \opt$,
proving the claim.
\end{proof}

\begin{proof}[Proof of Theorem~\ref{thr:lp_approx}]
Let $C$ represent the cost of the output returned by Algorithm~\ref{alg:rounding} and let $\opt$ denote the cost of the solution to \prbgtspalt. Similarly, let $\opt'$ be the cost of the ILP given in Eq.~\ref{eq:lp'_objective}, and let $\opt'_{\mathit{LP}}$ be the cost of the LP relaxation. 

Lemma~\ref{lem:relax} implies that
\[
    \opt'_{LP} \leq \opt' \leq (k - 1)\opt.
\]

Next, fix an edge $(u, v) \in E$ and consider the probability that $(u, v)$ becomes a backward edge after the rounding in Algorithm~\ref{alg:rounding}.
Let $B_i$ be the event where $v$ is assigned in a set $S_i$ while $u$ ends up in some later set, that is, $u \in S_{i'}$ with $i' > i$. The event $B_i$ occurs when
\[
    \sum_{j=1}^{i-1} x_v^j < \rho \leq \sum_{j=1}^i x_v^j
    \quad\text{and}\quad
    \sum_{j=1}^i x_u^j < \rho.
\]

Since $\rho$ is sampled uniformly, the probability of $B_i$ is the length of the intersection of the two corresponding intervals for $\rho$,
\begin{align*}
p(B_i) &= \max\fpr{\sum_{j=1}^i x_v^j - \max\fpr{\sum_{j=1}^{i-1} x_v^j,\  \sum_{j=1}^{i} x_u^j},0} \\
&= \max\fpr{\min\fpr{\sum_{j=1}^i x_v^j - \sum_{j=1}^{i-1} x_v^j,\  \sum_{j=1}^i x_v^j - \sum_{j=1}^{i} x_u^j}, 0}\\
&= \max\fpr{\min\fpr{x_v^i, p_v^i - p_u^i}, 0} \\
&= \min\fpr{\max\fpr{x_v^i, 0}, \max\fpr{p_v^i - p_u^i, 0}} \\
&= \min\fpr{x_v^i, \max\fpr{p_v^i - p_u^i, 0}},
\end{align*}
where $p_v^i$ is given in Eq.~\ref{eq:lp_p}. 

The probability that $(u, v)$ is a backward edge is then 
\begin{align*}
P((u,v) \text{ is a backward edge}) 
&= \sum_{i=1}^k P(B_i) \\
&= \sum_{i=1}^k \min(x_v^i, \max(p_v^i - p_u^i, 0)) \\ 
&\leq \sum_{i=1}^k \max(p_v^i - p_u^i, 0) = b'_{uv}.
\end{align*}

By a similar argument, we can see that the probability of a forward edge is at most $f_{uv}'$. Consider the probability that an edge $(u, v) \in E$ becomes a forward edge after the rounding in Algorithm~\ref{alg:rounding}. 

Let $F_i$ be the event that $v$ is assigned in set $S_i$ while $u$ ends up in some earlier set. The event $F_i$ occurs when
\[\sum_{j=1}^{i-1} x_v^j < \rho \leq \sum_{j=1}^i x_v^j\quad\text{and}\quad \rho \leq \sum_{j=1}^{i-1} x_u^j. \]

The probability of $F_i$ is the length of the corresponding interval for $\rho$,
\begin{align*}
P(F_i) &=\max\fpr{\min\fpr{\sum_{j=1}^i x_v^j, \sum_{j=1}^{i-1} x_u^j}- \sum_{j=1}^{i-1} x_v^j,0} \\
&= \max\fpr{\min\fpr{\sum_{j=1}^i x_v^j - \sum_{j=1}^{i-1} x_v^j, \sum_{j=1}^{i-1} x_u^j - \sum_{j=1}^{i-1} x_v^j}, 0} \\
&= \max\fpr{\min\fpr{x_v^i, p_u^{i-1} - p_v^{i-1}}, 0} \\
&= \max\fpr{\min\fpr{x_v^i, \fpr{1-p_v^{i-1}} - \fpr{1-p_u^{i-1}}}, 0} \\
&= \max\fpr{\min\fpr{x_v^i, q_v^i - q_u^i}, 0} \\
&= \min\fpr{\max\fpr{x_v^i, 0}, \max\fpr{q_v^i - q_u^i, 0}} \\
&= \min\fpr{x_v^i, \max\fpr{q_v^i - q_u^i, 0}},
\end{align*}
where $q_v^i$ is given in Eq.~\ref{eq:lp_q}. 

The probability that $(u, v)$ is a forward edge is then 
\begin{align*}
P((u,v) \text{ is a forward edge}) 
&= \sum_{i=1}^k P(F_i) \\
&= \sum_{i=1}^k \min(x_v^i, \max(q_v^i - q_u^i, 0)) \\ 
&\leq \sum_{i=1}^k \max(q_v^i - q_u^i, 0) = f'_{uv}.
\end{align*}

Finally, note that the probability that a vertex $u$ is assigned to cluster $i$ is \[P(u \text{ is assigned to } S_i) = \sum_{j=1}^i x_u^j - \sum_{j=1}^{i-1} x_u^j = x_u^i. \]

Therefore, the expected costs for the forward and backward edges, as well as the expected costs from the $L_2$ error term in the solution by Algorithm~\ref{alg:rounding}, are at most the corresponding costs for the LP relaxation, since
\begin{align*}
E[C] = & 
\sum_{u \in V} \sum_{i=1}^k P(u \text{ is assigned to } S_i) c_{u}^i  \\
&+ \sum_{e \in E} P(e \text{ is a forward edge}) \lambda_f + P(e \text{ is a backward edge}) \lambda_b  \\
\leq  & \sum_{u \in V} \left( \sum_{i=1}^k x_{u}^i c_{u}^i \right) +
\sum_{(u,v) \in E} \left( f'_{uv}\lambda_f + b'_{uv}\lambda_b \right) = \opt'_{LP}.
\end{align*}

Thus, 
\[
E[C] \leq \opt'_{LP} \leq \opt' \leq (k-1) \opt,
\]
which means that Algorithm~\ref{alg:rounding} yields a $k-1$ approximation in expectation.
\end{proof}

We can derandomize the rounding phase by testing every possible interval for $\rho$.

\begin{corollary}
There is a deterministic algorithm that yields a $k - 1$ approximation for \prbgtspalt that runs in $\bigO{T + k(nd + m)}$ time, where $T$ is the time required to solve LP with $\bigO{k(n + m)}$ variables and constraints.
\end{corollary}

\begin{proof}
Let $D = \set{\sum_{j=1}^{i-1} x_v^j \mid i = 1, \ldots, k, v \in V}$. Then for any $\rho' \in [0, 1]$, there is $\rho \in D$ such that $\rho$ and $\rho'$ yield the same clustering. Therefore, testing every $\rho \in D$ yields an approximation guarantee.

To obtain the running time, we consider each $\rho \in D$ in increasing order.
As we increase $\rho$, assume that a node $v$ changes its cluster. We can update the cost of the new assignment using the old cost in $\bigO{d + \deg v}$ time. As we increase $\rho$, each node can only change its cluster at most $k - 1$ times. This leads to a running time $\bigO{k \sum_{v \in V} d + \deg v} = \bigO{k(nd + m)}$ for the rounding phase.
\end{proof}

\subsection{Linear program approach for a case when $\lambda_f = \lambda_b$}

Next, we consider the symmetric case, where the weights for forward and backward edges are equal. In this case, the ordering of the groups no longer matters. We will show that using a random order for the groups yields a slightly better approximation guarantee of $\frac{k+1}{3}$. 

When $\lambda_f = \lambda_b = \lambda$, the objective for \prbgtspalt in Eq.~\ref{eq:lp_objective} simplifies to 
\begin{equation}
\label{eq:symmetric_lp_objective}
\textsc{Minimize} \quad \sum_{u \in V} \left( \sum_{i=1}^k x_{u}^i c_{u}^i \right) + 
\lambda \sum_{(u,v) \in E} d_{uv}.
\end{equation}
Here $d_{uv} = f_{uv} + b_{uv}$.
Moreover, Constraints~\ref{eq:lp_p_start}--\ref{eq:lp_d_start} can be simplified to
\begin{align*}
y_{uv}^i &\geq x_u^i - x_v^i, & i \in \{1,\ldots, k\}\\
y_{uv}^i &\geq x_u^i - x_v^i, & i \in \{1,\ldots, k\} \\
d_{uv} &= \frac{1}{2} \sum_{i=1}^k y_{uv}^i, & (u,v) \in E,
\end{align*}
where auxiliary variables $y_{uv}^i$ represent the difference $\abs{x_u^i-x_v^i}$ indicating that exactly one of $u$ and $v$ belongs to group $i$

Using ideas for an approximation algorithm for the multiway cut problem presented by~\citet{vazirani2002approximation}, we get the following lemma.

\begin{lemma}
\label{lemma:differ_2_coords}
Given two non-negative vectors $x$ and $y$ of length $k$ and $\norm{x}_1 = \norm{y}_1 = 1$, there is a sequence $Z = (z_1 = x, \ldots, z_k = y)$ of non-negative vectors with $\norm{z_i}_1 = 1$ such that $z_i$ and $z_{i + 1}$ differ in at most two entries and
\[
    \norm{x - y}_1 = \sum_{i} \norm{z_i - z_{i + 1}}_1.
\]
\end{lemma}

Throughout the proof, we will write $x^a$ to be the $a$th entry of a vector $x$.

\begin{proof}
Let $r$ denote the index for which the difference $\abs{x^r - y^r}$ is minimized. Without loss of generality, we may assume $x^r > y^r$, and let $\Delta = x^r - y^r$. Due to minimality of $r$ and the fact that $x$ and $y$ sum to 1, there must be another index $j$ such that $y^j \geq x^j + \Delta$.

Define $z$ to be $x$, except set $z^r = y^r = x^r - \Delta$ and 
$z^j = x^j + \Delta$. It follows that $\norm{z}_1 = 1$ and $\norm{x - y}_1 = \norm{x - z}_1 + \norm{y - z}_1$. Moreover, $x$ and $z$ differ only in two entries. We can repeat the process inductively for $z$ and $y$ to obtain the path $Z$.
The sequence is at most of length $k$ since the number of disagreements between $z$ and $y$ is smaller than the number of disagreements between $x$ and $y$. We can further assume that $Z$ contains $k$ elements by appending it with copies of $y$.
\end{proof}

\begin{algorithm}[t!]
\caption{LP rounding algorithm for the symmetric $\lambda_f = \lambda_b$ case}
\label{alg:rounding2}
Solve the LP relaxation for Equation~\ref{eq:symmetric_lp_objective} to obtain variables $x_v^i$\;
Sample uniformly a random permutation $\sigma : [1,k] \rightarrow [1,k]$\;
Sample uniformly a variable $\rho$ from the interval $[0, 1]$\;
\For{$v \in V$}{
    assign $v$ to the set $S_{\sigma(i)}$ where $i$ satisfies $\sum_{j=1}^{i-1} x_v^{\sigma(j)} < \rho \leq \sum_{j=1}^i x_v^{\sigma(j)}$\;
}
\Return sets $S_1, \ldots, S_k$\;
\end{algorithm}

\begin{theorem}
Algorithm~\ref{alg:rounding2} provides a $\frac{k+1}{3}$-approximation for \prbgtspalt in expectation.
\end{theorem}

\begin{proof}
Note that the probability that a vertex $u$ is assigned to group $i$ is \[P(u \text{ is assigned to } S_i) = \sum_{j=1}^i x_u^j - \sum_{j=1}^{i-1} x_u^j = x_u^i.
\]
Due to the linearity of expectation, the expected cost from the $L_2$ term matches the first term in Equation~\ref{eq:symmetric_lp_objective}.

To prove the claim, we will show that the expected probabilty of $u$ and $v$ ending up in different clusters is bounded by $\frac{k + 1}{3} d_{uv}$. The linearity of expectation then proves the result.

Consider the probability that for an edge $(u,v) \in E$, the vertices $u$ and $v$ are placed in different groups. If $x_u^i = x_v^i$ for all indices $i$, then $u$ and $v$ are always placed in the same group.

Apply Lemma~\ref{lemma:differ_2_coords} to $x_u$ and $x_v$ to obtain a sequence of vectors $Z$.
For notational simplicity, let us write
\[
    c_a(i) = \sum_{j=1}^{i} z_a^{\sigma(j)}
\]
to be the cumulative sum of the shuffled $z_a$.

Define $q_a = s$ to be the index $s$ such that
\[
    c_a(s-1) < \rho \leq c_a(s).
\]
Note that $u$ is assigned to $S_{\sigma(q_1)}$ 
while $v$ is assigned to $S_{\sigma(q_k)}$.

Fix $a < k$. Lemma~\ref{lemma:differ_2_coords} states that $z_a$ and $z_{a + 1}$ differ for exactly two indices, say $\sigma(i)$ and $\sigma(j)$ with $\sigma(i) < \sigma(j)$.

Assume $z_a^{\sigma(i)} < z_{a + 1}^{\sigma(i)}$. It follows that
\[
    c_{a + 1}(s) =
    \begin{cases}
    c_{a}(s) + \Delta, & \text{when } i \leq s \leq j, \\
    c_a(s), & \text{otherwise}, \\
    \end{cases}
\]
where $\Delta = 0.5\norm{z_a - z_{a + 1}}_1$.

Consequently $q_a \geq q_{a + 1}$. Given an index $s$,
the probability (conditioned on the permutation $\sigma$) that $q_a = s$ and $s \neq q_{a + 1}$ is bounded by
\begin{align*}
P(q_a = s, q_{a + 1} < s \mid \sigma) & = P(c_a(s-1) < \rho \leq c_a(s),\  \rho \leq c_{a + 1}(s - 1) \mid \sigma) \\
& \leq P(c_a(s-1) < \rho \leq c_{a + 1}(s - 1) \mid \sigma) \\
&= c_{a + 1}(s - 1)- c_a(s-1). \\
\end{align*}
The bound is $0$, when $s \leq i$ or $s > j$, and $\Delta$ if $i < s \leq j$.
Consequently,
\begin{align*}
P(q_a \neq q_{a + 1} \mid \sigma) & = \sum_{s = {i + 1}}^{j} P(q_a = s, q_{a + 1} < s \mid \sigma) \leq (j - i)\Delta. \\
\end{align*}

Since the indices are uniformly shuffled, $j - i - 1$, the expected number of indices properly between $i$ and $j$ is known to be $\frac{k-2}{3}$, and the probability that $q_a \neq q_{a + 1}$ is at most 
\[
P(q_a \neq q_{a + 1}) = \operatorname{E}_\sigma\fspr{P(q_a \neq q_{a + 1} \mid \sigma)} \leq
\fpr{\frac{k-2}{3} + 1} \Delta = \frac{k+1}{3} \Delta.
\]
The case $z_a^{\sigma(i)} > z_{a + 1}^{\sigma(i)}$ is analogous, completing the proof.
\end{proof}

We can turn the expected approximation guarantee into a high probability statement by performing the rounding process in Algorithm~\ref{alg:rounding2} several times.

\begin{corollary}
Repeating the rounding process $\ceil{- \log(\delta)/ \log (1 + \epsilon)}$ times and selecting the best solution yields a $\frac{k+1}{3}(1 + \epsilon)$ approximation with probability $1 - \delta$ or higher, for any $\epsilon, \delta > 0$.
\end{corollary}

\begin{proof}
Markov's inequality implies that a single rounding process achieves an approximation ratio worse than $\frac{k+1}{3}(1 + \epsilon)$ with a probability of at most $1/(1 + \epsilon)$. Repeating the rounding at least $\ceil{- \log(\delta)/ \log (1 + \epsilon)}$ times reduces the failure probability to $\delta$.
\end{proof}
\section{Additional heuristics for solving directed graph segmentation}\label{sec:algo}

In this section, we consider three additional algorithms \prbgts.
The first two algorithms are based on an iterative approach, that is,
we solve \prbgtspalt on fixed centroids followed by updating and reordering the centroids after each iteration as in Algorithm~\ref{alg:iter}.
The first algorithm utilizes the dynamic programming algorithm for tree graphs on the spanning forest of the input graph. The second algorithm is based on the $k = 2$ case, where we iteratively select pairs $i < j$ and optimize $S_i$ and $S_j$ while keeping everything else fixed. The final third algorithm is a greedy search where we update partitions by moving one node at a time.

\textbf{Dynamic programming on spanning forests.}
While the dynamic programming algorithm in Section~\ref{sec:sub} only works for tree graphs (and forests), we adapt it to general graphs by first computing a maximum-weight spanning forest on the undirected version of the input graph and then applying the dynamic programming algorithm. The pseudocode for this algorithm, called \algdp, is shown in Algorithm~\ref{alg:treedp}. Intuitively, \algdp may perform well if the undirected version of the graph does not contain many cycles and not too many edges are removed to create the spanning forest.

\begin{algorithm}[t!]
\caption{$\algdp(G = (V,E), \text{ initial partition } S_1, \ldots, S_k, \lambda_f, \lambda_b)$,\\ iterative dynamic programming algorithm on a maximum spanning forest}
\label{alg:treedp}
compute the maximum-weight spanning forest of the input graph\;
\While{the loss is decreasing or until a set amount of iterations} {
    solve \prbgtspalt on the spanning forest, and update the centroids\;
    sort clusters to minimize backward (forward) edges if $\lambda_b > \lambda_f$ ($\lambda_f > \lambda_b$)\;
}
\Return sets $S_1, \ldots, S_k$\;
\end{algorithm}

\textbf{Iterative two-group search.}
Our second approach, given
in Algorithm~\ref{alg:mincut}, is based on the special case for $k = 2$.
We iterate over all pairs $1 \leq i < j \leq k$ and for each
pair, we optimize $S_i$ and $S_j$ while keeping the remaining groups fixed
and all the centroids fixed. We will refer to this problem as $\prbgtspalt(i, j)$
Once $S_i$ and $S_j$ are updated, we update
the centroids $\mu_i$ and $\mu_j$.

\begin{algorithm}[t!]
\caption{$\algcut(G = (V,E), \text{ initial partition } S_1, \ldots, S_k, \lambda_f, \lambda_b)$,\\ iterative local search based on a minimum cut}
\label{alg:mincut}
\While{the loss is decreasing or until a set amount of iterations} {
    \ForEach {pair $i$, $j$ with $1 \leq i < j \leq k$} {
        solve $\prbgtspalt(i, j)$, and
        update centroids $\mu_i$ and $\mu_j$\;
    }
    sort clusters to minimize backward (forward) edges if $\lambda_b > \lambda_f$ ($\lambda_f > \lambda_b$)\;
}
\Return sets $S_1, \ldots, S_k$\;
\end{algorithm}

Solving $\prbgtspalt(i, j)$ is almost the same as solving \prbgtspalt for $k = 2$. 
The only main difference is that we need to take into account the cross edges from $S_i$ and $S_j$ to other groups. More formally,
we construct the same graph $H$ as in Section~\ref{sec:sub} except we set
the costs
\[
\begin{split}
    c(s, v) & = \norm{v - \mu_j}^2 + \lambda_f \abs{E(W, v)} + \lambda_b \abs{E(v, W)}, \ \text{and} \\
    c(v, t) & = \norm{v - \mu_i}^2 + \lambda_b \abs{E(W, v)} + \lambda_f \abs{E(v, W)},
    \ \text{where}\  W = S_{i + 1} \cup \cdots \cup S_{j - 1}.
\end{split}
\]

The next result implies that a minimum cut in $H$ solves
$\prbgtspalt(i, j)$.

\begin{theorem}
Let $C_1 = S_i \cup \set{s}$ and $C_2 = S_j \cup \set{t}$. 
Then $q(\fm{S}, \set{\mu_t})$ is equal to the total weight of the edges of $C_1$ to $C_2$ in $H$.
\end{theorem}

The proof is similar to the proof of Theorem~\ref{thr:mincut} and is therefore omitted.

As in the case $k = 2$, solving the minimum cut can be done deterministically in $\bigO{nm}$ time, and constructing $H$ requires $\bigO{nd + m}$ time, where $d$ is the length of the feature vector. Consequently, since there are $k(k - 1)/2$ pairs of $i, j$, a single iteration of the algorithm requires $\bigO{k^2n(d + m)}$ time.

\textbf{Greedy local search.}
As a third algorithm (see Algorithm~\ref{alg:greedy}), we consider a greedy approach where we start with
a partition and try to improve it by moving individual nodes from one group to another until convergence.

\begin{algorithm}[t!]
\caption{$\alggreedy(G = (V,E), \text{ initial partition } S_1, \ldots, S_k, \lambda_f, \lambda_b)$,\\ greedy local search} 
\label{alg:greedy}
\While{the loss is decreasing or until a set amount of iterations} {
    \ForEach{$v \in V$}{
        find optimal $S_j$ for $v$\;
        move $v$ from its current set $S_i$ to $S_j$, and
        update centroids $\mu_i$ and $\mu_j$\;
    }
    sort clusters to minimize backward (forward) edges if $\lambda_b > \lambda_f$ ($\lambda_f > \lambda_b$)\;
}
\Return sets $S_1, \ldots, S_k$\;
\end{algorithm}

Finding the optimal centroid for each vertex is expensive if we were to compute the cost from scratch. Luckily, we can compute the gain more efficiently, which leads to a faster running time.

\begin{theorem}
A single iteration of Algorithm~\ref{alg:greedy} requires $\bigO{k(nd + m)}$ time, where $d$ is the length of the feature vectors.
\label{thr:greedytime}
\end{theorem}

For each vertex $v$, we compute the change in $L_2$ error when moving $v$ from one group to another. We can do this using the following lemma.

\begin{lemma}
\label{lem:diff}
Let $\fm{S}$ be a partition. Select $i$ and $v \in S_i$.
Select $j$ and let $\fm{S}'$ be the result of moving $v$ from $S_i$ to $S_j$.
Let $\set{\mu_t}$ and $\set{\mu_t'}$ be the corresponding optimal centroids.
Let $H(S_i)$ be the $L_2$ error of $S_i$ and write 
$H(\fm{S}) = \sum H(S_i)$. Then
\[
    H(\fm{S}') - H(\fm{S}) = \abs{S_i}\norm{\mu_i}^2 + \abs{S_j}\norm{\mu_j}^2
    - \abs{S_i'}\norm{\mu_i'}^2 - \abs{S_j'}\norm{\mu_j'}^2 
\]
\end{lemma}

\begin{proof}
The identity
    $\sum_{w \in S_t} \norm{a(w) - \mu_t}^2 = (\sum_{w \in S_t} \norm{a(w)}^2) - \abs{S_t}\norm{\mu_t}^2$
immediately proves the claim.
\end{proof}

\begin{proof}[Proof of Theorem~\ref{thr:greedytime}]
Computing the gain of moving $v$ from $S_i$ to $S_j$ requires computing $\mu_i'$ and $\mu_j'$ which can be done in $\bigO{d}$ time using $\mu_i$ and $\mu_j$. Using Lemma~\ref{lem:diff} for the $L_2$ error and computing the changes in the cost terms for the edges in \bigO{\deg(v)} time allows us to compute the total cost difference in $\bigO{d + \deg(v)}$ time. Summing over $v$ and $j$ leads to a running time of $\bigO{k(nd + m)}$.
\end{proof}
\section{Experimental evaluation}\label{sec:exp}
In this section, we describe our experiments to test the algorithms in practice. We evaluate the algorithms first on synthetically constructed graphs and then using six real-world datasets. The experiments were run on a server equipped with an Intel Xeon Gold 6248 processor.
Data is publicly available, and we share our source code in a public repository online.\!\footnote{\url{https://version.helsinki.fi/dacs/coherent-groups-network}}

We used the following algorithms:

\alggreedy, \algdp, and \algcut from Section~\ref{sec:algo}.

\alglpiter and \alglpsym: These are iterative versions of Algorithm~\ref{alg:rounding} and Algorithm~\ref{alg:rounding2} by combining them with Algorithm~\ref{alg:iter} to update centroids between iterations.
We implement these algorithms using Gurobi Optimizer~\cite{gurobi}.
We exclude \alglpsym from experiments where $\lambda_f \neq \lambda_b$. While the Constraint~\ref{eq:lp_x_ne} is not required to solve \prbgtspalt and does not guarantee non-empty clusters with these algorithms, we include it in our implementations to reduce the number of empty clusters.

\algilpiter: This algorithm iteratively solves the exact integer programming formulation for \prbgtsp (Eq.~\ref{eq:lp_objective}), combining with Algorithm~\ref{alg:iter} to update the centroids.

\algexact: Mixed integer linear programming solver for \prbgts directly, as given in Section~\ref{sec:milp}.

\textbf{Baselines.}
We compare our algorithms with the following baselines.

\algkmeans: $k$-means clustering algorithm iteratively assigning vertices to nearest clusters and updating the centroids to minimize the $L_2$ errors~\cite{lloyd1982least}. Note that $k$-means ignores edges.

\algkcut: This baseline computes an approximate minimum-weight $k$-cut using Gomory--Hu trees to minimize the costs of edges between clusters~\cite{gomory1961multi}. Note that this method does not take $L_2$ errors into account. The vertices are first ordered topologically, and edges are weighted by $\lambda_f$ or $\lambda_b$ based on the vertex ordering. If the graph contains cycles, an approximate minimum-weight feedback arc set of edges is first removed using the greedy algorithm by~\citet{eades1993fast}. Suppose cycles remain due to the feedback arc set not being exact. In that case, the vertices are sorted by the topological order of strongly connected components in the remaining graph. In contrast, the order of vertices within each strongly connected component is arbitrary.

\algrandom: Returns a random partition of the vertices into $k$ clusters.

\textbf{Experiments on synthetic data.}
To test our algorithms, we create two types of synthetic graphs: tree graphs and directed acyclic graphs (DAGs). Each graph consists of $n$ vertices with $d$-dimensional features that separate them into $k$ clusters. For tree graphs, we have the vertices numbered from $1$ to $n$, and for each vertex $v$ except the first, we randomly add an edge from one of the earlier vertices to $v$. For DAGs, we start with the same tree and additionally randomly add an edge between each pair of vertices $v_i$, $v_j$ with $i<j$ with probability $0.05$.

We create the features by first sampling $k$ centroids uniformly distributed from $[0, 1]^{d}$. We then assign $v_{(n/k)(i - 1) + 1}, \ldots, v_{(n/k)i}$ to a cluster $S_i$ with the $i$th centroid. Each node then gets a normally distributed feature centred at the centroid of the cluster it belongs to, with a variance of $0.1$ in each dimension.

\textbf{Runtime scaling experiments.}
To assess the scalability of our algorithms as a function of the size of the inputs, we track the running time of each algorithm as a function of the number of vertices, the number of features, and the number of clusters on synthetic DAGs in Figure~\ref{fig:synthetic_runtime}. We use a time limit of 30 minutes and take the average runtime across 5 runs. For the runtime experiments, we use random initial partitions for the algorithms and use a greedy method for reordering the centroids between iterations~\cite{eades1993fast}.

As default values, we set the number of vertices to $n=200$, the number of clusters to $k=5$, the number of features to $d=5$, and $\lambda_f = \lambda_b = 1.0$. For each runtime experiment, we then vary the number of vertices, features, or clusters accordingly. The number of edges is $O(n^2)$ but with a small coefficient, since edges are added with a probability of $0.05$ for each pair of vertices to turn a random tree graph into a DAG.

\algexact is too slow to be run on graphs larger than around 10 vertices, while \algilpiter exceeds the 30-minute time limit between 160 and 320 vertices. The baselines \algrandom and \algkmeans are the fastest, taking less than one second for graphs with over 5\,000 vertices. \algdp and \alggreedy are also fast, taking around 10-15 seconds. The results agree with the theoretical analysis that \algcut, \alggreedy, and \algdp scale linearly with the size of the input graph.

From Figure~\ref{fig:synthetic_runtime}, we also see that the number of features has less impact on the running time. All algorithms scale in the same manner as they compute the loss of the final partition in the same way. Interestingly, linear programming algorithms seem to run faster with a moderate number of features than with a small number of features.

Finally, when varying the number of clusters $k$, we set the number of vertices as $1000$, since it is not possible to have more clusters than vertices. The runtime of \algdp and \alggreedy grows linearly with the number of clusters, while the runtime of \algcut scales quadratically. \alglpiter slows down significantly, taking more than 30 minutes with $k=40$, whereas the baseline algorithms are barely affected by the number of clusters.

\begin{figure}[t!]
{\includegraphics[]{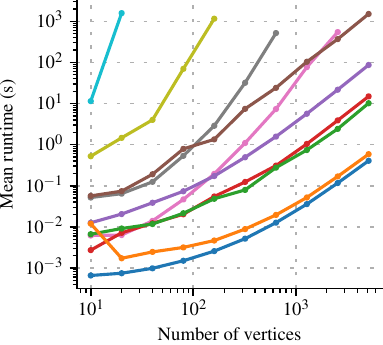}}
{\includegraphics[]{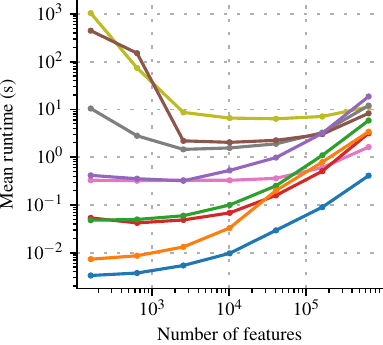}}\\
{\includegraphics[]{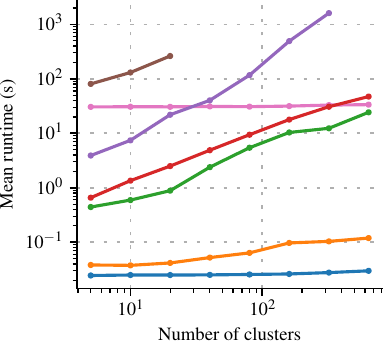}}
\hspace{0.2cm}
\raisebox{2.5cm}{
{\includegraphics[]{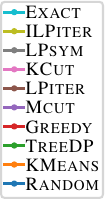}}
}
\caption{Algorithm average runtimes across 5 runs as a function of number of vertices (top left), number of features (top right), and number of clusters (bottom left) on synthetic graphs. Both axes are logarithmic.
}
\label{fig:synthetic_runtime}
\end{figure}

\textbf{Synthetic noise experiments.}
We create synthetic tree graphs and DAGs as before, with $n=1000$, $k=5$, and $d=10$. We define the sets $S_1, \ldots, S_5$ as the ground truth, each containing $n/k = 200$ vertices with features randomly drawn around a centroid. This partition initially minimizes the $L_2$ loss within the sets, and there cannot be backward edges between sets but only forward edges from a set $S_i$ to a later set $S_j$ with $i < j$. By setting $\lambda_f = 0$, the ground truth is then an optimal partition at the beginning with a loss of $0$. We then add random noise by, with a probability $p$, independently reassigning each node a new normally distributed feature around a random centroid.

We compare the similarity between the ground truth partition $\mathcal{S}$ and the partition $\mathcal{S'}$ returned by our algorithms by computing the Adjusted Rand Index, $\mathit{ARI} = \frac{\mathit{RI}-E[\mathit{RI}]}{\max(\mathit{RI}) - E[RI]}$. Here $\mathit{RI}(\mathcal{S}, \mathcal{S'}) = \frac{a+b}{{\binom{n}{2}}}$ is the Rand Index, where $a$ is the number of pairs of elements that are in the same set in both $\mathcal{S}$ and $\mathcal{S'}$, and $b$ is the number of pairs of elements that are in different sets in both $\mathcal{S}$ and $\mathcal{S'}$~\citep{hubert1985comparing}.

Since the \algkmeans algorithm runs very fast, we decide to use the output of \algkmeans as a starting partition to initialize the other algorithms. This reduces the number of iterations the other algorithms need to perform, and can improve the results, especially for \alggreedy, which may get stuck in bad local optima depending on the initial partition.

In Figure~\ref{fig:synthetic_ARI}, we plot $\mathit{ARI}$ between the output partition and the ground truth as a function of the probability $p$ of assigning new features to nodes. We set $\lambda_b = 1000$ and compute the average $\mathit{ARI}$ over 10 randomly generated graphs for each value of $p$. 

On both tree graphs and DAGs, $\mathit{ARI}$ values for \alglpiter, \algcut, \alggreedy, \algdp, and \algkmeans start from around 0.95, indicating that without noise, they are often able to find a partition that is similar to or exactly matches the ground truth. As the features become increasingly more random, $\mathit{ARI}$ for \algkmeans falls to 0, while the partitions by the other algorithms remain somewhat similar to the initial ground truth partition. This kind of partitioning is more suitable for practical applications, where the features for individual nodes are noisy or unreliable, and the network structure is more important.

For tree graphs, \algdp and \alglpiter consistently achieve slightly higher $\mathit{ARI}$ values than other algorithms. For DAGs, \alglpiter clearly outperforms others by finding partitions that closely match the ground truth, even when a large fraction of nodes have unreliable features. Surprisingly, the partitions returned by \algkcut are no more similar to the ground truth than the random partitions.

\begin{figure}[t!]
\centering
\hspace*{0.7cm}
\includegraphics[]{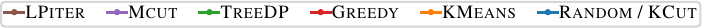}\\
\subcaptionbox{Synthetic trees}{\includegraphics[]{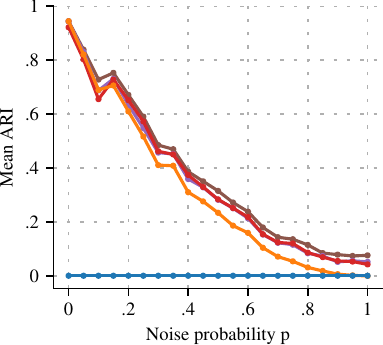}}
\subcaptionbox{Synthetic DAGs}{\includegraphics[]{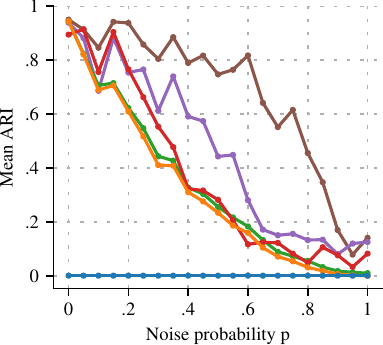}}
\caption{Mean Adjusted Rand Index over 10 runs between the ground truth and the partition chosen by the algorithms as a function of the probability $p$ of reassigning vertices with new random features from a random cluster. On synthetic trees, the values for \alglpiter and \algdp are overlapping.}
\label{fig:synthetic_ARI}
\end{figure}

\textbf{Experiments on real-world data.}
We perform experiments on the following six real-world datasets. The dataset sizes are shown in Table~\ref{tab:dataset_sizes}.

\begin{table}[t!]
\centering
\caption{Dataset sizes: number of vertices, edges, and features per vertex for each dataset.}
\label{tab:dataset_sizes}
\begin{tabular}{lrrr}
\toprule
\textbf{Dataset} & $\abs{V}$ & $\abs{E}$ & $d$ \\
\midrule
\dtenron             & 24\,277 & 28\,105 & 768 \\
\dtdblp              & 30\,581 & 70\,972 & 768 \\
\dtreddit            & 74\,778 & 74\,777 & 768 \\
\dtreddithyperlink   & 27\,863 & 121\,415 & 68 \\
\dttwitter           & 3\,377  & 162 & 768 \\
\dtwiki              & 4\,604  & 119\,882 & 768 \\
\bottomrule
\end{tabular}
\end{table}

\dtenron: We use the Enron Email Dataset\footnote{ \url{https://www.cs.cmu.edu/~./enron/}}, which consists of publicly available emails between employees of the former Enron Corporation. We create a graph where each node represents an email address, and each email creates directed edges from the sender to all recipients. If there are multiple emails between two nodes, we use a weighted edge, where the weight represents the number of emails from one node to another. Here, it would perhaps be more natural to assign features to edges rather than nodes, but since our problem requires features for the nodes, we give features to nodes as follows: We use an embedding model to convert email texts to vectors, and as the feature for each node, we assign the mean of the sentence embedding vectors from all email text contents sent from that email address. We exclude nodes for email addresses that do not send any emails.

\dtdblp: For the \dtdblp,\!\footnote{\url{https://www.aminer.org/citation}} dataset~\citep{Tang:08KDD}, we chose publications from ECMLPKDD, ICDM, KDD, NIPS, SDM, and WWW conferences. For each publication, we add an edge to those publications that cite it. To create the features for each node, we convert the titles of the scientific publications into sentence embeddings that we use as the feature vectors. 

\dtreddit: We use a popular Reddit thread on \emph{/r/politics}, which discusses the 2020 U.S. presidential election between Joe Biden and Donald Trump.\!\footnote{\url{https://www.reddit.com/r/politics/comments/jptq5n/megathread_joe_biden_projected_to_defeat/}} We use the Reddit API to collect the comments in the thread and construct a tree graph where the initial submission and each comment are nodes, and add an edge from each node to the comments responding to it. 

\dtreddithyperlink: This dataset\footnote{\url{https://snap.stanford.edu/data/soc-RedditHyperlinks.html}} represents a network between different communities on Reddit~\cite{kumar2018community}. Each node represents a subreddit, and posts containing a hyperlink to another post in different subreddits create directed edges between the two communities. In case of multiple posts, the edges are weighted by the number of posts linking from one subreddit to another. For each post, we create a 68-dimensional vector using the provided sentiment values and binary LWIC features that encode properties such as keywords and the topic categories of the post. For each node, we then assign a feature that is the mean of the corresponding vectors for each post from that subreddit.

\dttwitter: We used the former Twitter API to collect 10\,000 tweets that were posted by 27th of May, 23:59 UTC with the hashtag \#metoo. We then exclude retweets, and create a network where each node represents a user, and a tweet mentioning other users creates directed edges between the poster and the mentioned users. We convert the tweet texts to embedding vectors, and each user gets a feature that is the mean embedding of the tweets sent by that user. Since most tweets do not mention other users, the resulting graph is sparse and consists of multiple disconnected components.

\dtwiki: As the final dataset, we use the Wikispeedia network,\!\footnote{\url{https://snap.stanford.edu/data/wikispeedia.html}} which is a condensed version of the English language Wikipedia. Each node is a Wikipedia article, and directed edges represent links from one article to another. We create the features for each node from the article text content.

For all datasets except \dtreddithyperlink, we obtain the feature vectors from text by using a pretrained MPNet~\citep{song2020mpnet} language model to convert the text for each node into a $768$-dimensional vector. These embeddings aim to map texts with similar meanings to vectors that have a small distance in the feature space. We chose the all-mpnet-base-v2 language model for creating the feature vectors as it was conveniently accessible and its sentence embeddings achieved the best performance for general-purpose tasks out of the original HuggingFace sentence-transformer models~\cite{reimers2022sbert}. In future applications, this could be replaced by newer state-of-the-art embedding models. However, MPNet suffices for our experiments.

We set the number of clusters to $k=5$ and set $\lambda_f=0.01$ and $\lambda_b=0.1$. We use the output from the fast \algkmeans algorithm as an initial partition for the other algorithms. The results are shown in Table~\ref{tab:results}. The \algexact and \algilpiter algorithms are not included as they are too slow to be run on these datasets within a time limit of 8 hours. Since we set $\lambda_f \neq \lambda_b$ for this experiment, we did not run \alglpsym.

\begin{table}[p!]
\centering
\caption{Algorithm results for each dataset with $\lambda_f = 0.01, \lambda_b = 0.1$ and $k=5$. Each cell shows the mean $\pm$ standard deviation for running time, loss, and iterations across 10 runs. The algorithms with the lowest loss are shown in bold.}
\label{tab:results}
\begin{tabular}{llrrr}
\toprule
\textbf{Dataset} & \textbf{Algorithm} & \textbf{Time (s)} & \textbf{Loss} & \textbf{Iters} \\
\midrule
\multirow{7}{*}{\dtdblp}
  & \textbf{\alggreedy}  & 20 $\pm$ 3.4  & \textbf{24003 $\pm$ 39.7}  & 12 $\pm$ 2.3 \\
  & \algdp  & 6 $\pm$ 0.2  & 24882 $\pm$ 51.4  & 1 $\pm$ 0.0 \\
  & \algcut  & 120 $\pm$ 17.5  & 24169 $\pm$ 45.5  & 1 $\pm$ 0.0 \\
  & \alglpiter  & 2998 $\pm$ 680.2  & 24005 $\pm$ 35.6  & 32 $\pm$ 7.2 \\
  & \algkmeans  & 8 $\pm$ 2.3  & 24963 $\pm$ 80.7  & 89 $\pm$ 28.2 \\
  & \algkcut  & 5475 $\pm$ 292.9  & 25614 $\pm$ 0.0  & 1 $\pm$ 0.0 \\
  & \algrandom  & 0 $\pm$ 0.0  & 28824 $\pm$ 32.0  & 1 $\pm$ 0.0 \\
\midrule
\multirow{7}{*}{\dtenron}
  & \alggreedy  & 13 $\pm$ 3.7  & 17006 $\pm$ 6.5  & 12 $\pm$ 3.4 \\
  & \algdp  & 3 $\pm$ 0.1  & 17567 $\pm$ 31.9  & 1 $\pm$ 0.0 \\
  & \algcut  & 72 $\pm$ 24.9  & 17003 $\pm$ 12.7  & 1 $\pm$ 0.3 \\
  & \textbf{\alglpiter}  & 1810 $\pm$ 997.2  & \textbf{16986 $\pm$ 9.5}  & 42 $\pm$ 24.4 \\
  & \algkmeans  & 5 $\pm$ 0.8  & 17656 $\pm$ 73.7  & 68 $\pm$ 14.6 \\
  & \algkcut  & 2193 $\pm$ 100.3  & 18939 $\pm$ 0.0  & 1 $\pm$ 0.0 \\
  & \algrandom  & 0 $\pm$ 0.0  & 20137 $\pm$ 22.5  & 1 $\pm$ 0.0 \\
\midrule
\multirow{7}{*}{\dtreddit}
  & \alggreedy  & 36 $\pm$ 4.1  & 60089 $\pm$ 0.6  & 10 $\pm$ 1.1 \\
  & \textbf{\algdp}  & 42 $\pm$ 11.9  & \textbf{60069 $\pm$ 0.5}  & 8 $\pm$ 2.3 \\
  & \algcut  & 190 $\pm$ 74.8  & 60089 $\pm$ 6.6  & 1 $\pm$ 0.4 \\
  & \alglpiter  & 6734 $\pm$ 179.4  & 60137 $\pm$ 0.6  & 41 $\pm$ 0.3 \\
  & \algkmeans  & 14 $\pm$ 4.2  & 62332 $\pm$ 935.4  & 63 $\pm$ 23.5 \\
  & \algkcut  & 14764 $\pm$ 690.4  & 65622 $\pm$ 0.0  & 1 $\pm$ 0.0 \\
  & \algrandom  & 1 $\pm$ 0.0  & 69009 $\pm$ 642.2  & 1 $\pm$ 0.0 \\
\midrule
\multirow{7}{*}{\dtreddithyperlink}
  & \textbf{\alggreedy}  & 22 $\pm$ 0.3  & \textbf{2905 $\pm$ 0.0}  & 12 $\pm$ 0.0 \\
  & \algdp  & 28 $\pm$ 0.2  & 7613 $\pm$ 0.0  & 1 $\pm$ 0.0 \\
  & \algcut  & 145 $\pm$ 12.4  & 3459 $\pm$ 0.0  & 1 $\pm$ 0.0 \\
  & \alglpiter  & 1774 $\pm$ 37.9  & 2992 $\pm$ 0.0  & 14 $\pm$ 0.0 \\
  & \algkmeans  & 1 $\pm$ 0.1  & 8872 $\pm$ 933.7  & 33 $\pm$ 7.2 \\
  & \algkcut  & 6621 $\pm$ 726.1  & 8912 $\pm$ 0.0  & 1 $\pm$ 0.0 \\
  & \algrandom  & 0 $\pm$ 0.0  & 20086 $\pm$ 210.1  & 1 $\pm$ 0.0 \\
\midrule
\multirow{7}{*}{\dttwitter}
  & \textbf{\alggreedy}  & 1 $\pm$ 0.1  & \textbf{1620 $\pm$ 4.8}  & 4 $\pm$ 0.9 \\
  & \algdp  & 1 $\pm$ 0.2  & 1621 $\pm$ 4.7  & 3 $\pm$ 0.9 \\
  & \textbf{\algcut}  & 2 $\pm$ 0.2  & \textbf{1620 $\pm$ 4.9}  & 1 $\pm$ 0.0 \\
  & \textbf{\alglpiter}  & 10 $\pm$ 2.7  & \textbf{1620 $\pm$ 4.8}  & 5 $\pm$ 1.3 \\
  & \algkmeans  & 0 $\pm$ 0.1  & 1622 $\pm$ 4.9  & 22 $\pm$ 11.2 \\
  & \algkcut  & 18 $\pm$ 0.2  & 1932 $\pm$ 0.0  & 1 $\pm$ 0.0 \\
  & \algrandom  & 0 $\pm$ 0.0  & 1940 $\pm$ 0.8  & 1 $\pm$ 0.0 \\
\midrule
\multirow{7}{*}{\dtwiki}
  & \textbf{\alggreedy}  & 16 $\pm$ 2.7  & \textbf{4128 $\pm$ 9.9}  & 15 $\pm$ 2.6 \\
  & \algdp  & 7 $\pm$ 0.2  & 7398 $\pm$ 40.5  & 1 $\pm$ 0.0 \\
  & \algcut  & 130 $\pm$ 21.7  & 4134 $\pm$ 30.7  & 2 $\pm$ 0.3 \\
  & \alglpiter  & 462 $\pm$ 160.4  & 4161 $\pm$ 18.4  & 4 $\pm$ 2.3 \\
  & \algkmeans  & 1 $\pm$ 0.1  & 6959 $\pm$ 484.1  & 36 $\pm$ 12.3 \\
  & \algkcut  & 851 $\pm$ 74.1  & 4170 $\pm$ 0.0  & 1 $\pm$ 0.0 \\
  & \algrandom  & 0 $\pm$ 0.0  & 9448 $\pm$ 115.8  & 1 $\pm$ 0.0 \\
\bottomrule
\end{tabular}
\end{table}

Unsurprisingly, \algdp performs best on \dtreddit, which is a tree graph, and worse on graphs that do not resemble trees or forests. \alggreedy, \algcut, and \alglpiter are all competitive algorithms that consistently outperform the baselines. Most notably, \alggreedy achieves the lowest loss on four out of the six datasets while taking less than 40 seconds on all graphs. 

We observe that the algorithms require relatively few iterations to converge to a solution. In particular, \algcut typically converges after one or two iterations, whereas \alglpiter sometimes takes over 40 iterations, which can significantly increase the total running time.

\textbf{Case study: Wikipedia network}
To analyze whether our algorithms can identify sensible partitions, we examine the results on the \dtwiki dataset more closely. We set the number of clusters as $k=10$, $\lambda_f=0.001$, and $\lambda_b = 0.01$, and run the \alggreedy algorithm starting from a random initial partition.

The sizes of the groups in the resulting partition are shown in Table~\ref{tab:wiki_partition}, along with the titles of five articles closest to the centroid of each group. The articles around the centroids represent a wide range of coherent concepts, such as arts, locomotives, birds, and chemical elements.

While minimizing the $L_2$ error, the resulting partition also tries to avoid edges between groups, and orders the groups to restrict the number of backward edges. To this end, around $39\%$ of the articles are placed in the last group such that there are few edges towards earlier groups. Out of the total $119\,882$ edges, $29\%$ were forward edges, $8\%$ were backward edges, with the rest $63\%$ being edges within groups. In comparison, in a partition returned by \algkmeans, the order of the groups is random, and only $44\%$ of the edges are within groups.

\begin{table}[ht!]
\caption{Ordered partition of the \dtwiki dataset into 10 groups by \alggreedy, with $\lambda_f = 0.001, \lambda_b=0.01$. Each row shows the number of Wikipedia articles in the group and the titles of 5 articles in the group that are closest to the centroid.}
\begin{tabular*}{\textwidth}{@{}r r p{0.75\textwidth}@{}}
\toprule
\textbf{Group} & \textbf{Size} & \textbf{Articles} \\
\midrule
1 & 624 & Expressionism, Neoclassicism, Andy\_Warhol, Samuel\_Beckett, Novel \\
2 & 367 & 4-4-0, 4-6-0, Control\_car\_\%28rail\%29, EMD\_GP30, M-10003-6 \\
3 & 375 & History\_of\_Anglo-Saxon\_England, Gallery\_of\_the\_Kings\_and\_Queens\_of\_England, Peterborough\_Chronicle, Oliver\_Cromwell, Roman\_Britain \\
4 & 246 & \raggedright Natural\_disaster, Flood, Tropical\_cyclone, Meteorological\_history\_of\_Hurricane\_Katrina, Tsunami \\
5 & 280 & Bird, Blackbird, Coot, Sparrowhawk, Sparrow \\
6 & 220 & Sesame, Vegetable, Cultivar, Seed, Citrus \\
7 & 305 & \raggedright Astronomy, 3\_Juno, Astrophysics\_Data\_System, Physical\_science, Solar\_System \\
8 & 194 & Dinosaur, Reptile, Pelycosaur, Plesiosaur, Sauropsid \\
9 & 205 & Thorium, Tellurium, Praseodymium, Technetium, Americium \\
10 & 1788 & List\_of\_countries, 17th\_century, 1st\_century, Romania, Turkey \\
\bottomrule
\end{tabular*}
\label{tab:wiki_partition}
\end{table}
\section{Concluding remarks}\label{sec:conclusions}

In this paper, we considered partitioning the nodes of a directed graph into an ordered set of coherent groups. The objective was to minimize $L_2$ loss within groups and avoid crossing edges between groups, especially edges from nodes in later groups to earlier groups. 

We formulated a mixed-integer linear program that can solve our problem in exponential time. To find solutions in polynomial time, we considered heuristic approaches that utilize an iterative algorithm, alternating between fixing the centroids and optimizing the partition, and vice versa. We showed that finding a partition is an \np-hard problem, but that we can find the partition exactly for tree graphs, or when the number of groups is $k=2$.

For the general case, we created versions of the iterative approach based on greedy local search, spanning forests, and iteratively optimizing the partition for two groups at a time, as well as a linear programming rounding algorithm that finds a $k-1$-approximation for the subproblem.

We performed experiments on both synthetic and real-world datasets to demonstrate that the algorithms are practical, outperform baselines, and find coherent node groups in a small number of iterations and linear running time.

While this paper applied the $L_2$ loss function to measure the distance between two real-valued feature vectors, our methods could be used with other types of distance measures. In particular, networks with categorical features for the nodes provide an interesting line for future work. Another direction for future work would be to consider graphs where the feature vectors are assigned to the edges rather than the nodes. Alternatively, better initialization methods for our algorithms could be considered.

\backmatter

\bibliography{references}

\end{document}